\documentclass[11pt,letterpaper]{article}
\pdfoutput=1

\usepackage{fullpage}
\usepackage{hyperref,microtype}
\usepackage{graphicx}
\usepackage{amsbsy,amsfonts,amsmath,amssymb,amsthm,enumitem,latexsym}
\usepackage[square,numbers,sort&compress]{natbib}
\usepackage[font=small,margin=1em]{caption}

\theoremstyle{plain}
\newtheorem{thm}{Theorem}
\newtheorem{lemma}[thm]{Lemma}

\theoremstyle{definition}
\newtheorem{definition}{Definition}

\theoremstyle{definition}
\newtheorem*{red-gadget}{Red gadget}
\newtheorem*{blue-gadget}{Blue gadget}

\newcommand{\Apx}[1]{$#1$-approx\-i\-mation}

\newcommand{\E}{\mathbf{E}}
\renewcommand{\Pr}{\mathbf{Pr}}

\newcommand{\out}{\text{\upshape\sffamily out}}
\newcommand{\white}{\textsf{white}}
\newcommand{\black}{\textsf{black}}
\newcommand{\red}{\textsf{red}}
\newcommand{\blue}{\textsf{blue}}

\newcommand{\local}{$\mathcal{LOC\mspace{-3mu}AL}$}

\newcommand{\Alg}{\ensuremath{\mathcal{A}}}

\newcommand{\RC}{\textsc{Recut}}
\newcommand{\BVC}{\textsc{2-VC}}

\newcommand{\mydoi}[2]{\href{http://dx.doi.org/#1}{#2}}
\newcommand{\myeprint}[2]{\href{http://arxiv.org/abs/#1}{#2}}
\newcommand{\myurl}[2]{\href{#1}{#2}}

\setitemize{label=--}
\setenumerate{label=(\alph*)}

\hypersetup{
    colorlinks=true,
    linkcolor=black,
    citecolor=black,
    filecolor=black,
    urlcolor=black,
    pdfauthor={Mika G{\"o}{\"o}s, Jukka Suomela},
    pdftitle={No Sublogarithmic-time Approximation Scheme for Bipartite Vertex Cover},
}

\begin{document}
% --------------------------------------------------------------------------- %

\title{\bf No Sublogarithmic-time Approximation Scheme for Bipartite Vertex Cover}
\author{Mika G\"o\"os and Jukka Suomela \\[1ex] Helsinki Institute for Information Technology HIIT, \\ Department of Computer Science, University of Helsinki}
\date{}
\maketitle

\begin{abstract}
K\"onig's theorem states that on bipartite graphs the size of a maximum matching equals the size of a minimum vertex cover. It is known from prior work that for every $\epsilon > 0$ there exists a \emph{constant-time} distributed algorithm that finds a \Apx{(1+\epsilon)} of a maximum matching on 2-coloured graphs of bounded degree. In this work, we show---somewhat surprisingly---that no \emph{sublogarithmic-time} approximation scheme exists for the dual problem: there is a constant $\delta > 0$ so that no randomised distributed algorithm with running time $o(\log n)$ can find a \Apx{(1+\delta)} of a minimum vertex cover on 2-coloured graphs of maximum degree 3. In fact, a simple application of the Linial--Saks (1993) decomposition demonstrates that this lower bound is tight.

Our lower-bound construction is simple and, to some extent, independent of previous techniques. Along the way we prove that a certain cut minimisation problem, which might be of independent interest, is hard to approximate locally on expander graphs.
\end{abstract}

% --------------------------------------------------------------------------- %
\section{Introduction} \label{sec:intro}

Many graph optimisation problems do not admit an exact solution by a fast distributed algorithm. This is true not only for most NP-hard optimisation problems, but also for problems that can be solved using sequential polynomial-time algorithms. This work is a contribution to the \emph{distributed approximability} of such a problem: the minimum vertex cover problem on bipartite graphs---we call it $\BVC$, for short.

Our focus is on \emph{negative} results: We prove an optimal (up to constants) time lower bound $\Omega(\log n)$ for a randomised distributed algorithm to find a close-to-optimal vertex cover on bipartite 2-coloured graphs of maximum degree $\Delta = 3$. In particular, this rules out the existence of a sublogarithmic-time approximation scheme for $\BVC$ on sparse graphs.

Our lower bound result exhibits the following features:
\begin{itemize}
\item The proof is relatively simple as compared to the strength of the result;
this is achieved through an application of \emph{expander graphs} in the lower-bound construction.
\item To explain the source of hardness for $\BVC$ we introduce a certain \emph{distributed cut minimisation problem}, which might have applications elsewhere.
\item Many previous distributed inapproximability results are based on the hardness of \emph{local symmetry breaking}. This is not the case here: the difficulty we pinpoint for $\BVC$ is in the semi-global task of \emph{gluing together} two different types of local solutions.
\item Our result states that \emph{K\"onig's theorem is non-local}---see Section~\ref{sec:konig}.
\end{itemize}

\subsection{The \texorpdfstring{\local}{LOCAL} model}

We work in the standard \local{} model of distributed computing~\citep{linial92locality,peleg00distributed}. As input we are given an undirected graph $G=(V,E)$. We interpret $G$ as defining a communication network: the nodes $V$ host processors, and two processors can communicate directly if they are connected by an edge. All nodes run the same distributed algorithm $\Alg$. The computation of $\Alg$ on $G$ starts out with every node $v\in V$ knowing an upper bound on $n=|V|$ and possessing a globally unique $O(\log n)$-bit identifier $\operatorname{ID}(v)$; for simplicity, we assume that $V\subseteq \{1,2,\ldots,\operatorname{poly}(n)\}$ and $\operatorname{ID}(v)=v$. Also, we assume that the processors have access to independent (and unlimited) sources of randomness. The computation proceeds in synchronous communication rounds. In each round, all nodes first perform some local computations and then exchange (unbounded) messages with their neighbours. After some $r$ communication rounds the nodes stop and produce local outputs. Here $r$ is the \emph{running time} of $\Alg$ and the output of $v$ is denoted $\Alg(G,v)$.

The fundamental limitation of a distributed algorithm with running time $r$ is that the output $\Alg(G,v)$ can only depend on the information available in the subgraph $G[v,r]\subseteq G$ induced on the vertices in the radius-$r$ neighbourhood ball
\[
B_G(v,r)=\{u\in V:\operatorname{dist}_G(v,u)\leq r\}.
\]
Conversely, it is well known that an algorithm $\Alg$ can essentially discover the structure of $G[v,r]$ in time $r$. Thus, $\Alg$ can be thought of as a function mapping $r$-neighbourhoods $G[v,r]$ (together with the additional input labels and random bits on $B_G(v,r)$) to outputs.

While the \local{} model abstracts away issues of network congestion and asynchrony, this only makes our \emph{lower-bound} result stronger.

\subsection{Our result}
Below, we concentrate on bipartite 2-coloured graphs $G$. That is, $G$ is not only bipartite (which is a global property), but every node $v$ is informed of the bipartition by an additional input label $c(v)$, where $c\colon V\to\{\white,\black\}$ is a proper 2-colouring of $G$.

\begin{definition} In the $\BVC$ problem we are given a 2-coloured graph $G=(G,c)$ and the objective is to output a minimum-size vertex cover of $G$.
\end{definition}

A distributed algorithm $\Alg$ computes a vertex cover by outputting a single bit $\Alg(G,v)\in\{0,1\}$ on a node $v$ indicating whether $v$ is included in the solution. This way, $\Alg$ computes the set $\Alg(G) := \{v \in V : \Alg(G,v) = 1 \}$. Moreover, we say that $\Alg$ computes an \emph{\Apx{\alpha}} of $\BVC$ if $\Alg(G)$ is a vertex cover of $G$ and
\[
|\Alg(G)| \leq \alpha \cdot \textsf{OPT}_G,
\]
where $\textsf{OPT}_G$ denotes the size of a minimum vertex cover of $G$.

Our main result is the following.
\begin{thm} \label{thm:main}
There exists a $\delta > 0$ such that no randomised distributed algorithm with running time $o(\log n)$ can find a \Apx{(1+\delta)} of $\BVC$ on graphs of maximum degree $\Delta = 3$.
\end{thm}

A matching time upper bound follows directly from the well-known network decomposition algorithm due to Linial and Saks~\citep{linial93low}.
\begin{thm}
For every $\epsilon > 0$ a \Apx{(1+\epsilon)} of $\BVC$ can be computed with high probability in time $O(\epsilon^{-1} \log n)$ on graphs of maximum degree $\Delta = O(1)$.
\end{thm}
\begin{proof}
The subroutine \textsl{Construct\_Block} in the algorithm of Linial and Saks~\citep{linial93low} computes, in time $r=O(\epsilon^{-1} \log n)$, a set $S\subseteq V$ with the following properties. Each component in the subgraph $G[S]$ induced by $S$ has \emph{weak diameter} at most $r$, i.e., $\operatorname{dist}_G(u,v) \leq r$ for each pair $u,v\in S$ belonging to the same component of $G[S]$. Moreover, they prove that, w.h.p.,
\[
|S| \geq (1-\epsilon)n.
\]

Let $C$ be a component of $G[S]$. Every node of $C$ can discover the structure of $C$ in time $O(r)$ by exploiting its weak diameter. Thus, every node of $C$ can internally compute \emph{the same} optimal solution of $\BVC$ on $C$. We can then output as a vertex cover for $G$ the union of the optimal solutions at the components together with the vertices $V\smallsetminus S$. This results in a solution of size at most
\[
\textsf{OPT}_{G[S]} + \epsilon n.
\]
But since $\textsf{OPT}_G \geq |E|/\Delta = \Omega(n)$ for connected $G$, this is a \Apx{(1+O(\epsilon))} of $\BVC$.
\end{proof}

\subsection{K\"onig duality} \label{sec:konig}

The classic theorem of K\"onig (see, e.g., Diestel \citep[\S2.1]{diestel05graph}) states that, on bipartite graphs, the size of a maximum matching equals the size of a minimum vertex cover. A modern perspective is to view this result through the lens of linear programming (LP) duality. The LP relaxations of these problems are \emph{the fractional matching problem} (primal) and \emph{the fractional vertex cover problem} (dual):
\begin{align*}
 \text{maximise}\quad& {\sum_{e\in E} x_e}
&\text{minimise}\quad& {\sum_{v\in V} y_v} \\
 \text{subject to}\quad&
	{\sum_{e:\,v\in e} x_e \leq 1,\enspace \forall v\in V}
&\text{subject to}\quad&
	{\sum_{v:\,v\in e} y_v \geq 1,\enspace \forall e\in E}\\
 &\mathbf{x}\geq \mathbf{0}
&&\mathbf{y}\geq \mathbf{0}
\end{align*}
In particular, on bipartite graphs, the above LPs do not have an integrality gap (see, e.g., Papadimitriou and Steiglitz \citep[\S13.2]{papadimitriou98combinatorial}): among the optimal feasible solutions are integral vectors $\mathbf{x} \in \{0,1\}^E$ and $\mathbf{y}\in\{0,1\}^V$ that correspond to maximum matchings and minimum vertex covers.

In the context of distributed algorithms, the following is known on (bipartite) \emph{bounded degree graphs}:
\begin{enumerate}[label=(\arabic*)]
\item \emph{Primal LP and dual LP admit local approximation schemes.} As part of their general result, Kuhn et al.~\citep{kuhn06price} give a strictly local \Apx{(1+\epsilon)} scheme for the above LPs. Their algorithms run in constant time independent of the number of nodes.
\item \emph{Integral primal admits a local approximation scheme.} \AA{}strand et al.~\citep{astrand10weakly-coloured} describe a strictly local \Apx{(1+\epsilon)} scheme for the maximum matching problem on 2-coloured graphs. Again, the running time is a constant independent of the number of nodes.
\item \emph{Integral dual does not admit a local approximation scheme.} The present work shows---in contrast to the above positive results---that there is no local approximation scheme for $\BVC$ even when $\Delta = 3$.
\end{enumerate}

\subsection{Related lower bounds}

There are relatively few independent methods for obtaining negative results for distributed approximation in the \local{} model. We list three main sources.

\paragraph{Local algorithms.}
Linial's~\citep{linial92locality} lower bound $\Omega(\log^* n)$ for $3$-colouring a cycle together with the Ramsey technique of Naor and Stockmeyer~\citep{naor95what} establish basic limitations on finding \emph{exact solutions} strictly locally in constant time. These impossibility results were later extended to finding \emph{approximate solutions} on cycle-like graphs by Lenzen and Wattenhofer~\citep{lenzen08leveraging} and Czygrinow et al.~\citep{czygrinow08fast}. A recent work~\citep{goos12local-approximation} generalises these techniques even further and proves that deterministic local algorithms in the \local{} model are often no more powerful than algorithms running on anonymous port numbered networks. For more information on this line of research, see the survey of local algorithms~\citep{suomela09survey}.

Here, the inapproximability results typically exploit the inability of a local algorithm to break local symmetries. By contrast, in this work, we consider the case where the local symmetry is already broken by a $2$-colouring.

\paragraph{KMW-bounds.} Kuhn, Moscibroda and Wattenhofer~\citep{kuhn04what,kuhn06price,kuhn10local} prove that any randomised algorithm for computing a constant-factor approximation of minimum vertex cover on general graphs requires time $\Omega(\sqrt{\log n})$ and $\Omega(\log \Delta)$. Roughly speaking, their technique consists of showing that a fast algorithm cannot tell apart two adjacent nodes $v$ and $u$, even though it is globally more profitable to include $v$ in the vertex cover and exclude $u$ than conversely.

The lower-bound graphs of Kuhn et al.\ are necessarily of unbounded degree: on bounded degree graphs the set of all non-isolated nodes is a constant factor approximation of a minimum vertex cover. By contrast, our lower-bound graphs are of bounded degree and they forbid close-to-optimal approximation of $\BVC$.

\paragraph{Sublinear-time centralised algorithms.} Parnas and Ron~\citep{parnas07approximating} discuss how a fast distributed algorithm can be used as \emph{solution oracle} to a centralised algorithm that approximates parameters of a sparse graph $G$ in sublinear time given a randomised query access to $G$. Thus, lower bounds in this model of computation also imply lower bounds for distributed algorithms. In particular, an argument of Trevisan (presented in \citep{parnas07approximating}) implies that computing a \Apx{(2-\epsilon)} of a minimum vertex cover requires $\Omega(\log n)$ time on $d$-regular graphs, where $d=d(\epsilon)$ is sufficiently large.

We note that $\BVC$ is easy to approximate in this model: Nguyen and Onak~\citep{nguyen08constant-time} give a centralised constant-time algorithm to approximate the size of a maximum matching on a graph $G$. If we are promised that $G$ is bipartite, the same algorithm approximates the size of $\BVC$ by K\"onig duality.

% --------------------------------------------------------------------------- %
\section{Deterministic lower bound} \label{sec:deterministic}

To best explain the basic idea of our lower bound result, we first prove Theorem~\ref{thm:main} for a toy model that we define in Section~\ref{ssec:toy}; in this model, we only consider a certain class of deterministic distributed algorithms in anonymous networks. Later in Section~\ref{sec:randomised} we will show how to implement the same proof technique in a much more general setting: randomised distributed algorithms in networks with unique identifiers.

In the present section, we find a source of hardness for $\BVC$ as follows. First, we argue that any approximation algorithm for the $\BVC$ problem also solves a certain cut minimisation problem called $\RC$. We then show that $\RC$ is hard to approximate locally, which implies that $\BVC$ must also be hard to approximate locally.

\subsection{Toy model of deterministic algorithms}\label{ssec:toy}

Throughout this section we consider deterministic algorithms $\Alg$ running in time $r=o(\log n)$ that operate on \emph{input-labelled anonymous networks} $(G,\ell)$, where $G=(V,E)$ and $\ell$ is a labelling of $V$. More precisely, we impose the following additional restrictions in the \local{} model:
\begin{itemize}
 \item The nodes of $G$ are not given random bits as input.
 \item The output of $\Alg$ is invariant under reassigning node identifiers. That is, if $G'=(V',E')$ is isomorphic to $G$ via a mapping $f\colon V'\to V$, then for $v\in V$,
\[
 \Alg(G,\ell,v) = \Alg(G',\ell\circ f, f^{-1}(v)).
\]
\end{itemize}
Put otherwise, the only symmetry breaking information we supply $\Alg$ is the radius-$r$ neighbourhood topology together with the input labelling---the nodes are anonymous and do not have unique identifiers.

We will also consider graphs $G$ that are \emph{directed}. In this case, the directions of the edges are merely additional symmetry-breaking information; they do not restrict communication.

\begin{figure}[b]
    \centering
    \includegraphics[page=1]{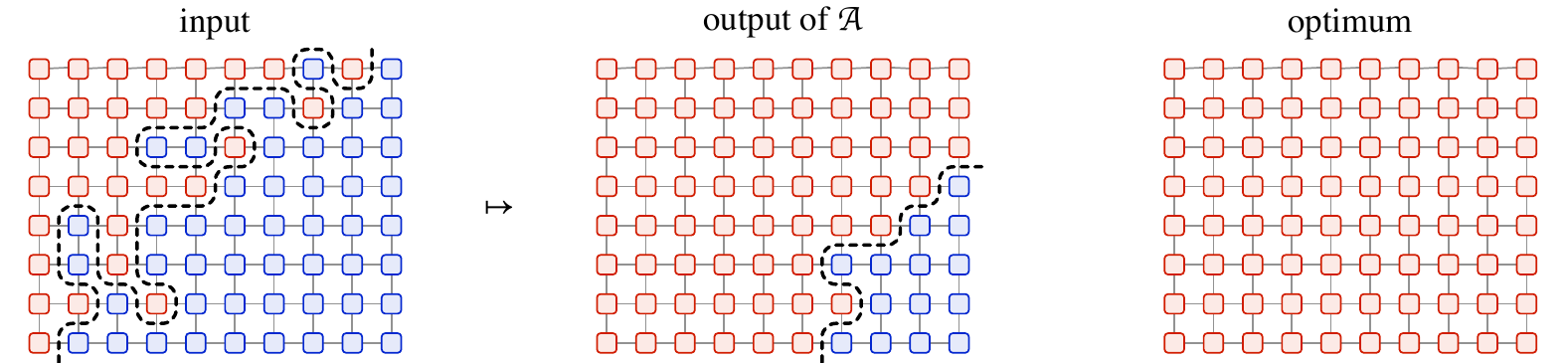}
    \caption{The $\RC$ problem. In this example, we have used a simple distributed algorithm~$\Alg$ to find a recut $\ell_\out$ with a small boundary $\partial\ell_\out$: a node outputs red iff there is a red node within distance $r = 3$ in the input. While the solution is not optimal, in a grid graph the boundary will be relatively small. However, our lower bound shows that any fast distributed algorithm---including algorithm~$\Alg$---fails to produce a small boundary in some graph.}\label{fig:recut}
\end{figure}

\subsection{The Recut problem}

In the following, we consider partitions of $V$ into red and blue colour classes as determined by a labelling $\ell\colon V\to\{\red,\blue\}$. We write $\partial\ell$ for the fraction of edges crossing the red/blue cut, i.e.,
\[
\partial\ell := \frac{e(\ell^{-1}(\red),\ell^{-1}(\blue))}{|E|}.
\]

\begin{definition} In \emph{the \RC{} problem} we are given a labelled graph $(G,\ell)$ as input and the objective is to compute an output labelling (\emph{a recut}) $\ell_\out$ that minimises $\partial\ell_\text{out}$ subject to the following constraints:
(a)~If $\ell(V) = \{\red\}$, then $\ell_\out(V) = \{\red\}$.
(b)~If $\ell(V) = \{\blue\}$, then $\ell_\out(V) = \{\blue\}$.
\end{definition}
In words, if we have an all-red input, we have to produce an all-red output, and if we have an all-blue input, we have to produce an all-blue output. Otherwise the output can be arbitrary. See Figure~\ref{fig:recut} for an illustration.

Needless to say, the global optimum for an algorithm $\Alg$ would be to produce a constant output labelling $\ell_\Alg$ (either all red or all blue) having $\partial\ell_\Alg = 0$. However, a distributed algorithm $\Alg$ can only access the values of the input labelling $\ell$ in its local radius-$r$ neighbourhood: when encountering a neighbourhood $v\in U\subseteq V$ with $\ell(U)=\{\red\}$, the algorithm is forced to output red at $v$ to guarantee satisfying the global constraint~(a), and when encountering a neighbourhood $v\in U\subseteq V$ with $\ell(U)=\{\blue\}$, the algorithm is forced to output blue at $v$ to satisfy~(b). Thus, if a connected graph $G$ has two disjoint $r$-neighbourhoods $U,U'\subseteq V$ with $\ell(U)=\{\red\}$ and $\ell(U')=\{\blue\}$ the algorithm $\Alg$ cannot avoid producing at least some red/blue edge boundary. Indeed, the best we can hope $\Alg$ to achieve is a recut $\ell_\Alg$ of size $\partial\ell_\Alg\leq \epsilon$ for some small constant $\epsilon>0$.

\paragraph{Discussion.}
The $\RC$ problem models the following abstract high-level challenge in designing distributed algorithms: Each node in a local neighbourhood $U\subseteq V$ can, in principle, internally compute a completely \emph{locally optimal} solution for (the subgraph induced by) $U$, but difficulties arise when deciding which of these proposed solution are to be used in the final distributed output. In particular, when the \emph{type} of the produced solution changes from one (e.g., red) to another (e.g., blue) across a graph $G$ one might have to introduce suboptimalities to the solution at the (red/blue) boundary in order to glue together the different types of local solutions.

In fact, the $\RC$ problem captures the first non-trivial case of this phenomenon with only \emph{two} solution types present. One can think of the input labelling $\ell$ as recording the \emph{initial preferences} of the nodes whereas the output labelling $\ell_\Alg$ records how an algorithm $\Alg$ decides to combine these preferences into the final unified output. In the end, our lower-bound strategy will be to argue that any $\Alg$ can be forced into producing too large an edge boundary $\partial\ell_\Alg$ resulting in too many suboptimalities in the produced output.

Next, we show how the above discussion is made concrete in the case of the $\BVC$ problem.

\subsection{The reduction}

We call a graph $G$ \emph{tree-like} if all the $r$-neighbourhoods in $G$ are trees, i.e., $G$ has girth larger than $2r+1$. Furthermore, if $G$ is directed, we say it is \emph{balanced} if $\operatorname{in-degree}(v) = \operatorname{out-degree}(v)$ for all vertices $v$. We note that a deterministic algorithm $\Alg$ produces the same output on every node of a balanced regular tree-like digraph $G$, because such a graph is \emph{locally homogeneous}: all the $r$-neighbourhoods of $G$ are pairwise isomorphic.

Using this terminology we give the following reduction.

\begin{thm} \label{thm:reduction}
Suppose $\Alg$ computes a \Apx{(1+\epsilon)} of $\BVC$ on graphs of maximum degree $\Delta = 3$. Then, there is an algorithm (with running time $r$) that finds a recut $\ell_\Alg$ of size $\partial\ell_\Alg = O(\epsilon)$ on balanced $4$-regular tree-like digraphs.
\end{thm}

The proof of Theorem~\ref{thm:reduction} follows the usual route: We give a local reduction (i.e., one that can be computed by a local algorithm) that transforms an instance $(G,\ell)$ of $\RC$ into an instance $\Pi(G,\ell)$ of $\BVC$. Then we simulate $\Alg$ on the resulting instance and map the output of $\Alg$ back to a solution of the $\RC$ instance $(G,\ell)$.

Let $G=(V,E)$ be a balanced $4$-regular tree-like digraph and let $\ell\colon V\to\{\red,\blue\}$ be a labelling of~$G$. The instance $\Pi(G,\ell)$ is obtained by replacing each vertex $v\in V$ by one of two local gadgets depending on the label $\ell(v)$. We first describe and analyse simple gadgets yielding instances of $\BVC$ with $\Delta = 4$; the gadgets yielding instances with $\Delta=3$ are described later.

\begin{red-gadget}
The red gadget replaces a vertex $v\in V$ by two new vertices $w_v$ (white) and $b_v$ (black) that share a new edge $e_v:=\{w_v,b_v\}$. The incoming edges of $v$ are reconnected to $w_v$, whereas the outgoing edges of $v$ are reconnected to $b_v$.
\end{red-gadget}

Note that the $\BVC$ instance $\Pi(G,\red)$ (where we denote by $\red$ the constant labelling $v\mapsto \red$) contains $\{e_v:v\in V\}$ as a perfect matching. Since $(G,\red)$ is locally homogeneous, the solutions output by $\Alg$ on the endpoints of $e_v$ are isomorphic across all $v$. Assuming $\epsilon < 1$ it follows that the algorithm $\Alg$ must output either the set of all white nodes or the set of all black nodes on $\Pi(G,\red)$. Our reduction branches at this point: we choose the structure of the blue gadget to counteract this white/black decision made by $\Alg$ on the red gadgets. We describe the case that $\Alg$ outputs all white nodes on $\Pi(G,\red)$; the case of black nodes is symmetric.

\begin{center}
    \includegraphics[page=2]{figs.pdf}
\end{center}

\begin{blue-gadget}
The blue gadget replacing $v\in V$ is identical to the red gadget with the exception that a third new vertex $w'_v$ (white) is added and connected to $b_v$.
\end{blue-gadget}

Similarly as above, we can argue that $\Alg$ outputs exactly the set of all black nodes on the instance $\Pi(G,\blue)$. This completes the description of $\Pi$.

Next, we simulate $\Alg$ on $\Pi(G,\ell)$. The output of $\Alg$ is then transformed back to a labelling $\ell_\Alg\colon V\to\{\red,\blue\}$ by setting
\[
\ell_\Alg(v)=\textsf{blue}\ \iff\ \text{the output of $\Alg$ contains only the black node $b_v$ at the gadget at $v$.}
\]
Note that $\ell_\Alg$ satisfies both feasibility constraints (a) and (b) of $\RC$. It remains to bound the size $\partial\ell_\Alg$ of this recut.

\begin{center}
    \includegraphics[page=3]{figs.pdf}
\end{center}

\paragraph{Recut analysis.} Call a red vertex $v$ in $(G,\ell_\Alg)$ \emph{bad} if $v$ has a blue out-neighbour $u$. By the definition of ``$\ell_\Alg(u)=\blue$'', the vertex cover produced by algorithm $\Alg$ does not contain the white node~$w_u$. Thus to cover the edge $(b_v,w_u)$, the vertex cover has to contain the black node~$b_v$. But by the definition of ``$\ell_\Alg(v)=\red$'', we must have $w_v$ or $w'_v$ in the solution as well. Hence, at least two nodes are used to cover the gadget at~$v$, which is suboptimal as compared to the minimum vertex cover $\{b_v:v\in V\}$, which uses only one node per gadget. This implies that we must have at most $\epsilon |V|$ bad vertices as $\Alg$ produces a \Apx{(1+\epsilon)} of $\BVC$ on $\Pi(G,\ell)$.

\begin{center}
    \includegraphics[page=4]{figs.pdf}
\end{center}

On the other hand, exactly half of the edges crossing the cut $\ell_\Alg$ are oriented from red to blue since $G$ is balanced. Each bad vertex gives rise to at most two of these edges, so we have that $\partial\ell_\Alg\cdot|E| / 2 \leq 2 \epsilon |V|$ which gives $\partial\ell_\Alg \leq 2\epsilon$, as required. This proves Theorem~\ref{thm:reduction} for $\Delta = 4$.

\paragraph{Gadgets for $\Delta = 3$.} The maximum degree used in the gadgets can be reduced to 3 by the following modification. The red gadget replaces a vertex $v\in V$ by a path of length 3.
\begin{center}
    \includegraphics[page=5]{figs.pdf}
\end{center}
Again, to achieve a \Apx{1.499} of $\BVC$ on $\Pi(G,\red)$ the algorithm $\Alg$ has to make a choice: either leave out the middle black vertex or the middle white vertex from the vertex cover. Supposing $\Alg$ leaves out the middle black, the blue gadget is defined to be identical to the red gadget with an additional white vertex connected to the middle black one.

After simulating $\Alg$ on an instance $\Pi(G,\ell)$ we define $\ell_\Alg(v)=\blue$ iff $\Alg$ outputs only black nodes at the gadget at $v$. The recut analysis will then give $\partial\ell_\Alg\leq 4\epsilon$.

\subsection{Recut is hard on expanders}

Intuitively, the difficulty in computing a small red/blue cut in the $\RC$ problem stems from the inability of an algorithm $\Alg$ to overcome the neighbourhood expansion of an input graph in $r=o(\log n)$ steps---an algorithm cannot hide the red/blue boundary as the radius-$r$ neighbourhoods themselves might have large boundaries.

To formalise this intuition, we use as a basis for our lower-bound construction an infinite family $\mathcal{F}$ of $4$-regular \emph{$\delta$-expander graphs}, where each $G=(V,E)\in\mathcal{F}$ satisfies the edge expansion condition
\begin{equation} \label{eq:expansion}
 e(S,V\smallsetminus S) \geq \delta\cdot |S|\qquad\text{for all}\enspace S\subseteq V,\enspace |S|\leq n/2.
\end{equation}
Here, $e(S,V\smallsetminus S)$ is the number of edges leaving $S$ and $\delta > 0$ is an absolute constant independent of $n=|V|$. On such graphs it is enough for us to force an algorithm $\Alg$ to output a \emph{nearly balanced recut} $\ell_\Alg$ having both colour classes close to $n/2$ in size. This is because if the number of, say, the red nodes is
\[
|\ell_\Alg^{-1}(\red)| = n/2 - o(n),
\]
then the expansion property \eqref{eq:expansion} implies that
\[
\partial\ell_\Alg \geq \delta/4 - o(1).
\]
That is, $\Alg$ computes a recut of size $\Omega(\delta)$.

Indeed, the following simple fooling trick makes up the very core of our argument.

\begin{lemma} \label{lem:fooling}
Suppose $\Alg$ produces a feasible solution for the $\RC$ problem in time $r=o(\log n)$. Then for each $4$-regular graph $G$ there exists an input labelling for which $\Alg$ computes a nearly balanced recut.
\end{lemma}
\begin{proof}
Fix an arbitrary ordering $v_1,v_2,\ldots,v_n$ for the vertices of $G$ and define a sequence of labellings
$\ell^0,\, \ell^1,\, \ldots,\, \ell^n$
by setting $\ell^i(v_j) = \blue$ iff $j\leq i$. That is, in $\ell^0$ all nodes are red, in $\ell^n$ all nodes are blue, and $\ell^i$ is obtained from $\ell^{i-1}$ by changing the colour of $v_i$ from red to blue.

When we switch from the instance $(G,\ell^{i-1})$ to $(G,\ell^i)$ the change of $v_i$'s colour is only registered by nodes in the radius-$r$ neighbourhood of $v_i$. This neighbourhood has size $|B_G(v_i,r)|\leq 5^r = o(n)$, and so the number of red nodes in the outputs $\ell^{i-1}_\Alg$ and $\ell^i_\Alg$ of $\Alg$ can only differ by $o(n)$. As, by assumption, we have that $\Alg$ computes the labelling $\ell^0_\Alg=\red$ on $(G,\ell^0)$ and the labelling $\ell^n_\Alg=\blue$ on $(G,\ell^n)$, it follows by continuity that some labelling in our sequence must force $\Alg$ to output $n/2-o(n)$ red nodes.
\end{proof}

We now have all the ingredients for the lower-bound proof: We can take $\delta = 2 - \sqrt{3}$ if we choose $\mathcal{F}$ to be the family of $4$-regular Ramanujan graphs due to Morgenstern~\citep{morgenstern94existence}. These graphs are tree-like, as they have girth $\Theta(\log n)$. They can be made into balanced digraphs since a suitable orientation can always be derived from an Euler tour. Thus, $\mathcal{F}$ consists of balanced $4$-regular tree-like digraphs. Lemma~\ref{lem:fooling} together with the discussion above imply that every algorithm for $\RC$ produces a recut of size $\Omega(\delta)$ on some labelled graph in $\mathcal{F}$. Hence, the contrapositive of Theorem~\ref{thm:reduction} proves Theorem~\ref{thm:main} for our deterministic toy algorithms.

% --------------------------------------------------------------------------- %
\section{Randomised lower bound} \label{sec:randomised}

Even though our model of deterministic algorithms in Section~\ref{sec:deterministic} is an unusually weak one, we can quickly recover the standard \local{} model from it by equipping the nodes with independent sources of randomness. In particular, as is well known, each node can choose an identifier uniformly at random from, e.g., the set $\{1,2,\ldots,n^3\}$, and this results in the identifiers being globally unique with probability at least $1-1/n$.

When discussing randomised algorithms many of the simplifying assumptions made in Section~\ref{sec:deterministic} no longer apply. For example, a randomised algorithm need not produce the same output on every node of a locally homogeneous graph. Consequently, the homogeneous feasibility constraints in the $\RC$ problem do not strictly make sense for randomised algorithms.

However, we can still emulate the same proof strategy as in Section~\ref{sec:deterministic}: we force the randomised algorithm to output a nearly balanced recut with high probability. Below, we describe this strategy in case of the easy-to-analyse ``$\Delta = 4$'' gadgets with the understanding that the same analysis can be repeated for the ``$\Delta = 3$'' gadgets with little difficulty.

\subsection{Repeating Section~\ref{sec:deterministic} for randomised algorithms}
Fix a randomised algorithm $\Alg$ with running time $r=o(\log n)$ and let $G=(V,E)$, $n=|V|$, be a large $4$-regular expander.

Again, we start out with the all-red instance. We denote by $W$ and $B$ the number of black and white nodes output by $\Alg$ on $\Pi(G,\red)$. As each of the edges $e_v$ must be covered, we have that
\[
W + B \geq n.
\]
Hence, by linearity of expectation, at least one of $\E[W] \geq n/2$ or $\E[B] \geq n/2$ holds. We assume that $\E[W] \geq n/2$; the other case is symmetric.

In reaction to $\Alg$ preferring white nodes, the blue gadgets are now defined exactly as in Section~\ref{sec:deterministic} with an additional white vertex. Furthermore, for any input $\ell\colon V\to\{\red,\blue\}$ we interpret the output of $\Alg$ on $\Pi(G,\ell)$ as defining an output labelling $\ell_\Alg$ of $V$, where, again, $\ell_\Alg(v)=\blue$ iff $\Alg$ outputs only the black node at the gadget at $v$. This definition translates our assumption of $\E[W] \geq n/2$ into
\begin{equation} \label{eq:r-red}
\E[R(\red)] \geq n/2,
\end{equation}
where $R(\ell) := |\ell_\Alg^{-1}(\red)|$ counts the number of gadgets (i.e., vertices of $G$) relabelled red by $\Alg$ on $\Pi(G,\ell)$.

If $\Alg$ relabels a blue gadget red, it must output at least two nodes at the gadget. This means that the size of the solution output by $\Alg$ on $\Pi(G,\blue)$ is at least $n + R(\blue)$. Thus, if $\Alg$ is to produce a \Apx{3/2} on $\Pi(G,\blue)$ in expectation, we must have that
\begin{equation} \label{eq:r-blue}
\E[R(\blue)] \leq n/2.
\end{equation}

The inequalities \eqref{eq:r-red} and \eqref{eq:r-blue} provide the necessary boundary conditions (replacing the feasibility constraints of $\RC$) for the argument of Lemma~\ref{lem:fooling}: by continuously changing the instance $(G,\red)$ into $(G,\blue)$ we may find an input labelling $\ell^*$ achieving
\begin{equation} \label{eq:expectation}
\E[R(\ell^*)] = n/2 - o(n).
\end{equation}

It remains to argue that $\Alg$ outputs a nearly balanced recut not only ``in expectation'' but also with high probability.

\subsection{Local concentration bound}

Focusing on the instance $\Pi(G,\ell^*)$ we write $R=R(\ell^*)$ and
\begin{equation} \label{eq:sum}
R = \sum_{v\in V} X_v, 
\end{equation}
where $X_v\in\{0,1\}$ indicates whether $\Alg$ relabels the gadget at $v$ red.

The variables $X_v$ are not \emph{too} dependent: the $2r$th power of $G$, denoted $G^{2r}$, where $u,v\in V$ are joined by an edge iff $B_G(v,r)\cap B_G(u,r) \neq \varnothing$, is \emph{a dependency graph} for the variables $X_v$. Every independent set $I\subseteq V$ in $G^{2r}$ corresponds to a set $\{X_v\}_{v\in I}$ of mutually independent random variables. Since the maximum degree of $G^{2r}$ is at most $\max_v |B_G(v,2r)| = o(n)$, this graph can always be partitioned into $\chi(G^{2r})=o(n)$ independent sets.

Indeed, Janson~\citep{janson04large} presents large deviation bounds for sums of type \eqref{eq:sum} by applying Chernoff--Hoeffding bounds for each colour class in a $\chi(G^{2r})$-colouring of $G^{2r}$. For any $\epsilon > 0$, Theorem 2.1 in Janson~\citep{janson04large}, as applied to our setting, gives
\begin{equation} \label{eq:concentration}
\Pr( R \geq \E[R] + \epsilon n) \leq \exp\left(-2\frac{(\epsilon n)^2}{\chi(G^{2r})\cdot n}\right) \to 0,\quad\text{as}\ n\to \infty,
\end{equation}
and the same bound holds for $\Pr\left( R \leq \E[R] - \epsilon n\right)$. That is, $R$ is concentrated around its expectation.

In conclusion, the combination of \eqref{eq:expectation} and \eqref{eq:concentration} implies that, for large $n$, the algorithm $\Alg$ outputs a nearly balanced recut on $\Pi(G,\ell^*)$ with high probability. By the discussion in Section~\ref{sec:deterministic}, this proves Theorem~\ref{thm:main}.

% --------------------------------------------------------------------------- %
\section{Acknowledgements}

Many thanks to Valentin Polishchuk for discussions. This work was supported in part by the Academy of Finland, Grants 132380 and 252018.

\setlength{\bibsep}{2ex minus 2ex}
\bibliographystyle{myplain}

\end{document}